\documentclass[aps,prl,reprint,superscriptaddress]{revtex4-1}

\usepackage{amssymb,amsmath}
\usepackage{graphicx}
\usepackage{amsthm}

\newtheorem{thm}{Theorem}

\def\imagetop#1{\vtop{\null\hbox{#1}}}

\def\bea{\begin{eqnarray}}
\def\eea{\end{eqnarray}}
\def\ben{\begin{equation}}
\def\een{\end{equation}}
\def\benu{\begin{enumerate}}
\def\enu{\end{enumerate}}

\def\lsim {\ifmmode {\buildrel<\over\sim}}

\def\sss{\scriptscriptstyle\rm}

\def\1var{(\bx_1...\bx\N)}

\def\br{{\bf r}}
\def\bR{{\bf R}}

\def\b1{{\bf 1}}
\def\bx{{x}}

\def\Hxc{_{\sss HXC}}

\def\N{_{\sss N}}
\def\H{_{\sss H}}

\def\sph_int{ {\int d^3 r}}

\def\infintd3r{ \int_{-\infty}^\infty d^3r\,}
\def\intd3r{ \int d^3r\,}

\def\laplace1d{\frac{d^2}{dx^2}}
\def\plaplace1d{\frac{d^2}{d{x'}^2}}

\def\padr2{\frac{\partial^2}{\partial r^2}}

\def\N{{\cal N}}
\def\a{{\alpha}}
\def\b{{\beta}}

\newcommand{\mr}[1]{\mathrm{#1}}
\newcommand{\mc}[1]{\mathcal{#1}}

\newcommand{\mb}[1]{\mathbf{#1}}

\begin{document}


\title{Fragment-based Time-dependent Density-functional Theory}


\author{Mart\'in A. Mosquera}
\affiliation{Department of Chemistry, Purdue University, West Lafayette, IN 47907, USA}
\author{Daniel Jensen}
\affiliation{Department of Physics, Purdue University, West Lafayette, IN 47907, USA}
\author{Adam Wasserman}
\email{Corresponding author.\\awasser@purdue.edu}
\affiliation{Department of Chemistry, Purdue University, West Lafayette, IN 47907, USA}
\affiliation{Department of Physics, Purdue University, West Lafayette, IN 47907, USA}
\thanks{Corresponding author}


\date{\today}

\begin{abstract}
Using the Runge-Gross theorem that establishes the foundation of Time-dependent Density Functional Theory (TDDFT) we prove that for a given electronic Hamiltonian, choice of initial state, and choice of fragmentation, there is a {\em unique} single-particle potential (dubbed {\em time-dependent partition potential}) which, when added to each of the pre-selected fragment potentials, forces the fragment densities to evolve in such a way that their sum equals the {\em exact} molecular density at all times.  This uniqueness theorem suggests new ways of computing time-dependent properties of electronic systems via {\em fragment}-TDDFT calculations. We derive a formally exact relationship between the partition potential and the total density, and illustrate our approach on a simple model system for binary fragmentation in a laser field.

\end{abstract}

\pacs{}
\maketitle

Time-dependent density functional theory (TDDFT) \cite{RG84,Ullrich-book} allows one to predict, in principle, the 
evolution of the non-relativistic density $n(\br t)$ of a system of
interacting electrons subject to a time-dependent external potential $v(\br t)$.
Given an initial wave function, the time-dependent electron density determines
the external potential up to a time-dependent constant (Runge-Gross theorem [RG] \cite{RG84}) and the density may be found by solving
the time-dependent Kohn-Sham (TDKS) equations. These equations make it possible to perform
practical calculations to propagate the electronic density and its related quantities such as linear response functions. Due to its wide range
of applications, TDDFT is expected to continue being a workhorse in the coming years for chemistry, physics, and materials engineering \cite{BWG05}.

Although the computational cost of TDKS calculations is low compared to that of other
many-body techniques, new ideas are needed to enable the study of
larger systems with improved efficiency and accuracy.
For the ground-state problem, `divide-and-conquer' fragmentation techniques have been developed \cite{Y91} and applied successfully through the use of readily-available parallel computers. Related strategies have also been developed recently for the time-dependent problem within TDDFT \cite{PJ12}. For example, 
\citet{CW04} introduced a methodology to perform time-dependent
calculations within frozen-density embedding theory. It has been shown that this
method yields better results than ``supermolecular'' techniques in some cases \cite{FLWP11,NLBW05}. 
Other extensions include linear-response TDDFT for molecules in solvents \cite{IMTC04} and
TDDFT for interacting chromophores \cite{HFG01}.
Also, time-dependent calculations within subsystem-DFT have been reported 
(see \cite{N10} and references therein).
In subsystem-DFT, the density of the system is split into densities of localized subsystems. Then the Kohn-Sham kinetic energy of the total 
system is truncated,
and the energy functional is approximated by a functional of the subsystem densities;
its minimization
leads to Kohn-Sham equations for each subsystem. Neugebauer formulated this theory within 
linear response in the frequency domain and showed that it yields results consistent
with conventional TDDFT \cite{N07}. For the case of  
dissipative dynamics, Zhou et al. \citep{ZTRG10} showed how the RG theorem
can be applied and Kohn-Sham equations developed for {\em open} systems given an initial state, memory kernel, and system-bath correlation.

Among density-based {\rm ground-state} fragmentation techniques, Partition Density Functional Theory (PDFT) \cite{EBCW10} is a reformulation of Density Functional Theory that allows one to find the solution to the KS equations without solving the total molecular problem directly. The idea is to partition the
external potential into an arbitrary number of fragment potentials. The total energy of the isolated systems is minimized 
under
the constraint that the fragment densities sum to the correct molecular density. The Lagrange multiplier associated with the constraint (i.e. the {\em partition potential}) can be found by inversion if the total density is known \cite{CW07}, or via the self-consistent procedure of Ref. \cite{EBCW10} if it is not.
Every fragment is subject to the same partition potential. In contrast with quantum mechanical embedding theories (except for the latest version of quantum embedding \cite{HPC11}) and
with subsystem-DFT,
this potential is global and unique \cite{CW06}. The set of fragment densities obtained for a given choice of external-potential
partitioning is also unique. As \citet{P13} recently suggested, this uniqueness feature of PDFT makes it a suitable
candidate to simplify the formulation of subsystem-DFT. This letter reports on foundational work for such developments. We extend PDFT to the time-dependent regime and show how
the time-dependent external field can be partitioned.
A new potential termed the {\em time-dependent partition potential} is introduced in the formalism
in order to represent the exact time-dependent electronic density. 

To extend PDFT to the time-dependent domain we recall that there is no minimum principle from which 
the TDKS equations can be derived \cite{R96,V08}. In view of this, we follow a deductive approach
to define our TDKS equations. 
Our goal is to provide a fragment-based solution to the Liouville equation (we use atomic units throughout)
\begin{equation}
i\frac{\partial\hat{\Gamma}\left(t\right)}{\partial t}=\left[\hat{H}_{v}(t),\hat{\Gamma}(t)\right].
\label{e:Liouville}
\end{equation}
If $\hat{\Gamma}$ is a pure density matrix then Eq. (\ref{e:Liouville}) is equivalent to the
time-dependent Schr\"{o}dinger equation. We suppose that the initial state $\hat{\Gamma}(t_0)$
is given. In standard DFT notation, the Hamiltonian is given by
$
\hat{H}_v(t)=\hat{T}+\hat{V}_{ee}+\int \text{d}^3\mb{r}~ v\left(\mb{r}t\right)\hat{n}\left(\mb{r}\right)
$.
It is convenient to express the external potential $v(\br t)$ as the sum of the potential $\tilde{v}(\br)$ due to the $M$ nuclei $\tilde{v}(\br)=-\sum_\a^M Z_\a/|\br-\bR_\a|$, which is not explicitly time-dependent, and an additional potential $v_E(\br t)$ containing all of the explicit time-dependence due to external fields:
\ben
v(\br t)=\tilde{v}(\br)+v_E(\br t)
\een 
Our task is to divide the quantum system into $N_f$ fragments of
interacting electrons. This is done by assigning an external potential $v_{\alpha}(\br t)$, Hamiltonian
$\hat{H}_{\alpha}(t)$, and initial state $\hat{\Gamma}_{\alpha}(t_0)$ to each fragment.
Out of the infinitely many ways to choose the fragment potentials,
there are at least two cases that are physically relevant: (i) Direct partitioning of the time-dependent external potential in analogy to ground-state PDFT:
$v(\mb{r}t)=\sum_{\alpha}^{N_f} v_{\alpha}(\mb{r}t)$.
For example, if $N_f=M$, there are cases of interest where we could define
$v_{\alpha}(\mb{r}t)=-Z_{\alpha}/\left|\mb{r}-\mb{R}_{\alpha}(t)\right|$.
In such cases, the electronic density of fragment $\alpha$ would be an output variable of the 
dynamics of nucleus $\alpha$. In general, however, we find option (ii) to be more convenient: Fragment the {\em static} potential only,
\begin{equation}
\tilde{v}(\br)=\sum_{\alpha}\tilde{v}_{\alpha}(\br)~~,
\end{equation}
and define the time-dependent fragment potential $v_\alpha(\br t)$ by adding the total time-dependent potential $v_E(\br t)$ to each of the $\tilde{v}_\alpha(\br t)$'s:
\begin{equation}
v_{\alpha}(\mb{r}t)=\tilde{v}_{\alpha}(\mb{r})+v_{E}(\mb{r}t)~.
\end{equation}
Now define the many-electron fragment-$\alpha$ Hamiltonian as
\begin{equation}
\hat{H}_{\alpha}(t)=\hat{T}+\hat{V}_{ee}+\int \text{d}^3\mb{r}\, 
\left[v_{\alpha}(\mb{r}t)+v_p(\mb{r}t)\right]\hat{n}(\mb{r})~.
\label{e:fragment_H}
\end{equation}
The evolution of the state of this particular
fragment is governed by the Liouville equation
\begin{equation}
i\frac{\partial}{\partial t}\,\hat{\Gamma}_{\alpha}(t)=\left[\hat{H}_{\alpha}(t),\hat{\Gamma}_{\alpha}(t)\right]~.
\end{equation}
The time-dependent electronic density of fragment $\alpha$ is given by
$
n_{\alpha}(\br t)=\mr{Tr}\{\hat{\Gamma}_{\alpha}(t)\hat{n}(\mb{r})\}
$
and the {\em time-dependent partition potential} $v_p(\br t)$ of Eq. (\ref{e:fragment_H}) is defined by requiring that the sum of fragment densities reproduce the total molecular density at all times:
\begin{equation}
\sum_{\alpha=1}^{N_f}n_{\alpha}(\mb{r}t)=n(\mb{r}t)~~.
\label{e:density_constrain}
\end{equation}

Just like traditional TDDFT is based on a one-to-one mapping between the Kohn-Sham
potential $v_s(\mb{r}t)$ and the electronic density $n(\br t)$, we now prove an analogous one-to-one mapping between $n(\br t)$ and $v_p(\br t)$. The latter is therefore sharply defined by  Eqs. (\ref{e:Liouville})-(\ref{e:density_constrain}). 
 
\begin{thm}
  For a given set of initial states $\{\hat{\Gamma}_{\alpha}(t_0)\}$, the map between 
  the density and the partition potential is invertible up to 
  a time-dependent constant in the potential.
\end{thm}
\begin{proof}
The proof uses the Runge-Gross theorem \cite{RG84}, and is analogous to it. Suppose there is a minimum integer $k\ge 0$ such that 
\begin{equation}
\frac{\partial^m}{\partial t^m}[v_p'(\mb{r}t)-v_p(\mb{r}t)]\Big|_{t=t_0}
\begin{cases}
=\mr{Constant}\quad  m<k\\
\neq \mr{Constant}\quad m=k~.
\end{cases}
\end{equation} 
Also assume that $v_p$ and $v_p'$ have the associated densities $\{n_{\alpha}\}$ and
$\{n_{\alpha}'\}$ correspondingly. 
Suppose $\hat{H}_{\alpha}(t)$ and $\hat{H}_{\alpha}'(t)$ are the Hamiltonians of fragment $\alpha$
that correspond to $v_p$ and $v_p'$, respectively. The key for the proof is the continuity equation
\begin{equation}
\frac{\partial n_{\alpha}(\mb{r}t)}{\partial t}=-\nabla \cdot \mb{j}_{\alpha}(\mb{r}t)
\end{equation}
and the Liouville equation for the fragment current densities
\begin{equation}
i\frac{\partial \mb{j}_{\alpha}(\mb{r}t)}{\partial t}=\mr{Tr}\left\{\hat{\Gamma}_{\alpha}(t)
\left[\hat{\mb{j}}(\mb{r}),\hat{H}_{\alpha}(t)\right]\right\}~.
\end{equation}
Define 
\begin{equation}
w_k^p(\mb{r})=\frac{\partial^{k}}{\partial t^{k}}\left[
v_p'(\mb{r}t)-v_p(\mb{r}t)\right]\Big|_{t=t_0}~.
\end{equation}
In virtue of the Runge-Gross theorem \cite{RG84} and its generalization to 
ensembles \cite{LT85}, it is easy to show that
\begin{equation}
\frac{\partial^{k+2}}{\partial t^{k+2}}\left[n'_{\alpha}(\mb{r}t)-n_{\alpha}(\mb{r}t)\right]|_{t=t_0}
=-\nabla \cdot \left[ n_{\alpha}(\mb{r}t_0)\nabla w_k^p(\mb{r}) \right]~.
\end{equation}
Summing over all fragments gives
\begin{equation}\label{intkey}
\frac{\partial^{k+2}}{\partial t^{k+2}}[n'(\mb{r}t)-n(\mb{r}t)]|_{t=t_0}
=-\nabla \cdot \left[ n(\mb{r}t_0)\nabla w_k^p(\mb{r}) \right]~.
\end{equation}
Now we show that the right-hand side of this equation cannot be zero.
Assume $\nabla \cdot [ n(\mb{r}t_0)\nabla w_k^p(\mb{r}) ]=0$ and $\nabla w_k^p\neq 0$. Now invoke Green's identity to find
\begin{equation}
\begin{split}
\int \text{d}^3\mb{r}\,w_k^p(\mb{r})&\nabla\cdot(n(\mb{r}t_0)\nabla w_k^p(\mb{r}))=
-\int \text{d}^3\mb{r} \,n(\mb{r}t_0)(\nabla w_k^p(\mb{r}))^2\\&
+\frac{1}{2}\oint \text{d}\mb{S}\cdot n(\mb{r}t_0)\nabla (w_k^p)^2(\mb{r})=0~.
\end{split}
\end{equation}
If the total electronic density falls off enough to make the surface term
negligible then $\nabla w_{k}^p(\mb{r})=0$, which is a contradiction. Therefore
the right-hand side of Eq. (\ref{intkey}) cannot be zero.
This leads to the conclusion that if $v_p'$ and $v_p$ differ
by more than a time-dependent constant then they cannot yield 
the same density in time.  
\end{proof}
The above theorem shows that if $\{\Gamma_{\alpha}(t_0)\}$ and $v_p(\br t)$ are given, 
then one obtains a {\em unique} set of fragment densities
$\{n_{\alpha}(\br t)\}$ and total density $n(\br t)$. The fragment density $n_{\alpha}(\br t)$ can be assumed to 
be non-interacting $v$-representable in time. Then we can associate a time-dependent Kohn-Sham
potential $v_{s,\alpha}(\br t)$ and initial state $\hat{\Gamma}_{s,\alpha}(t_0)$ to
describe the evolution of $n_{\alpha}(\br t)$ by means of the KS equations:
\begin{equation}
i\partial_t\varphi_{i\alpha}(\mb{r}t)=\left[-\frac{1}{2}\nabla^2+v_{s,\alpha}(\mb{r}t)\right]\varphi_{i\alpha}(\mb{r}t)~,
\end{equation}
where 
\begin{equation}
n(\mb{r}t)=\sum_{\alpha}n_{\alpha}(\mb{r}t)=\sum_{i\alpha}f_{i\alpha}\left|\varphi_{i\alpha}(\mb{r}t)\right|^2~.
\end{equation}

The fragments are allowed to have non-integer average numbers of electrons that are set by the initial state \cite{CW06}.
Since the Hamiltonian is particle-conserving,
the occupation numbers $f_{i\alpha}$ remain fixed during the propagation. 

In analogy with PDFT, we define the xc potential by means of
\begin{widetext}
\begin{equation}\label{vs}
v_{{\rm xc},\alpha}[n_{\alpha},\hat{\Gamma}_{\alpha}(t_0),\hat{\Gamma}_{s,\alpha}(t_0)](\mb{r}t)=
v_{s,\alpha}[n_{\alpha},\hat{\Gamma}_{s,\alpha}(t_0)](\mb{r}t)-
v_{\alpha}[n_{\alpha},\hat{\Gamma}_{\alpha}(t_0)](\mb{r}t)-
  v\H[n_{\alpha}](\mb{r}t)-v_p[n,\{\hat{\Gamma}_{\alpha}(t_0)\}](\mb{r}t)~.
\end{equation}
\end{widetext}
By comparing the fragment continuity equations for the interacting and non-interacting (Kohn-Sham) systems, we find that the above
definition of the xc potential is consistent with (for example, see \cite{GM12})
\begin{equation}
\begin{split}
\nabla\cdot\Bigg\{n_{\alpha}(\mb{r}t)&\nabla\Bigg[\int \text{d}^3\mb{r}\,\frac{n_{\alpha}(\mb{r}'t)}{|\mb{r}'-\mb{r}|}+\\&v_{{\rm xc},\alpha}
(\mb{r}t)\Bigg]\Bigg\}=\nabla\cdot[\mb{Q}_{\alpha}(\mb{r}t)-\mb{Q}_{s,\alpha}(\mb{r}t)]~,
\label{e:fragment_Q}
\end{split}
\end{equation}
where the right-hand sides are hydrodynamical terms given by 
$
\mb{Q}_{s,\alpha}(\mb{r}t)=-i\mr{Tr}\{ \Gamma_{s,\alpha}(t)[\,\hat{\mb{j}}(\mb{r}),\hat{T}]\}$ and
$\mb{Q}_{\alpha}(\mb{r}t)= -i\mr{Tr}\{ \Gamma_{\alpha}(t)[\,\hat{\mb{j}}(\mb{r}),\hat{T}+\hat{V}_{ee}]\}$.
This indicates that the conventional xc potential of TDDFT and family of approximations
can be used for the fragments' TDKS equations, a direct consequence of van Leeuwen's 
theorem \cite{L99}.

Furthermore, from the continuity equations for the total current and fragment current densities, and from Eqs. (\ref{e:fragment_Q}), we find a formally exact relationship between the time-dependent partition potential and the total density:
\begin{equation}
\begin{split}
\nabla \cdot(n(\mb{r}t)\nabla v_p(\mb{r}t))=
\frac{\partial^2 n(\mb{r}t)}{\partial t^2}+
\sum_{\alpha}\Big\{\nabla\cdot \mb{Q}_{s\alpha}[v_p](\mb{r}t)\\
-\nabla \cdot(n_{\alpha}[v_p](\mb{r}t)\nabla \bar{v}_{s,\alpha}[v_p](\mb{r}t))\Big\}
\label{e:inver}
\end{split}
\end{equation}
where $\bar{v}_{s,\alpha}[v_p]=v_{\Hxc,\alpha}[v_p]+v_{\alpha}$. In principle,
evaluation of Eq. (\ref{e:inver}) at $t=t_0$ yields a Sturm-Liouville linear differential equation where
$v_p(\mb{r},t=t_0)$ is the unknown variable. If we assume that the 
density is Taylor-expandable at $t=t_0$, then it is easy to show that
consecutive differentiation of Eq. (\ref{e:inver}) and evaluation
at $t=t_0$ leads to a family of equations from which the Taylor 
coefficients of $v_p(\br t)$ can be constructed in increasing order.
This suggests that a given density is $v_p$-representable as long
as the conditions of the Sturm-Liouville theory are met.

To illustrate our fragmentation approach, consider the simplest non-trivial model system consisting of a one-dimensional ``electron", two fragments, and an oscillating electric field of fixed frequency. For the static part of the external potential we choose a sum of soft-Coulomb potentials of equal strength $V_0$, a distance $l$ apart:
\begin{equation}
\tilde{v}(x)=V_0 \left(\frac{1}{\sqrt{(x+l/2)^2+a}}+\frac{1}{\sqrt{(x-l/2)^2+a}}\right).
\end{equation}
For the laser field we choose $v_{E}(x,t)=xE\sin(\omega t)$, with $E=0.1$, and $\omega=0.3$.

We partition the system by defining
$v_1(x,t)=V_0/\sqrt{(x+l/2)^2+a}+v_E(x,t)$ and $v_2(x,t)=V_0/\sqrt{(x-l/2)^2+a}+v_E(x,t)$.
The time-dependent fragment equations are, (for $\alpha=1,2$),
\begin{equation}
i\frac{\partial}{\partial t}\varphi_{\alpha}(xt)=
\left(-\frac{1}{2}\frac{\text{d}^2}{\text{d}x^2}+v_{\alpha}(xt)+v_p(xt)\right)\varphi_{\alpha}(xt)
\label{e:ptdks}
\end{equation}

\begin{figure}[h]
\centering
\begin{tabular*}{1.0\textwidth}{cc}
\imagetop{\includegraphics[width=0.24\textwidth]{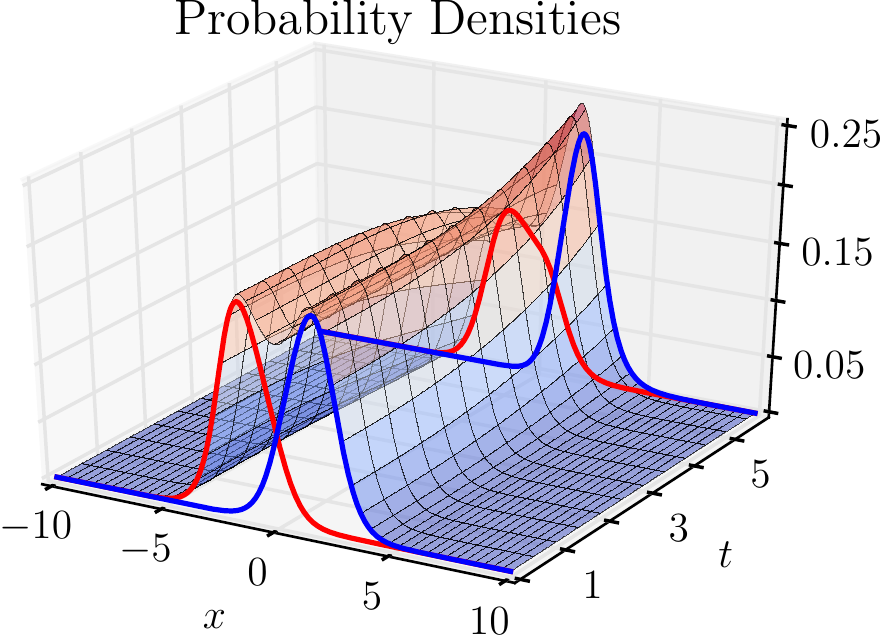}} & 
\imagetop{\includegraphics[width=0.24\textwidth]{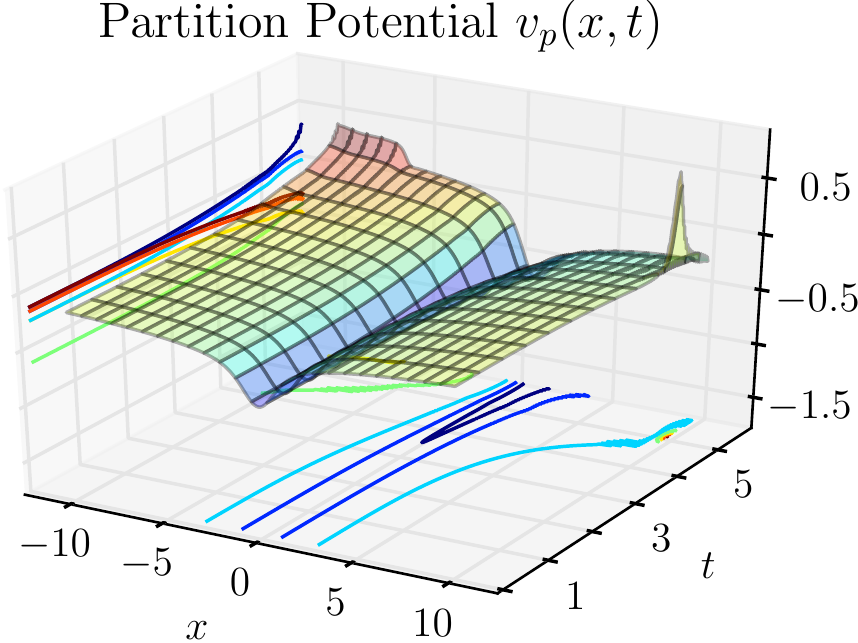}}
\end{tabular*}
\caption{{\em Left:} Fragment densities at times $t=0$ and $t\approx T/4$ along with the
total molecular density. {\em Right:} Exact partition potential $v_{p}\left(x,t\right)$. 
Parameters: $l = 1$, $a=1$, $V_{0}=-1$, $E=0.1$, and $\omega=0.3$. A box of length
20 with 401 points was used.}
\label{f:fig1}
\end{figure}

The initial states of the fragments are obtained by solving the ground-state 
PDFT equations as prescribed in Ref. \cite{MW12}. This procedure generates
the initial fragment Kohn-Sham orbitals needed to solve Eqs. \ref{e:ptdks}.
The distance between the wells was chosen
to allow for a significant overlap between the initial fragments' densities. 



Even though the principle to construct $v_p(\br t)$ is simple, note that
Eq. (\ref{e:inver}) can also be written as $v_p=\mc{F}v_p$, where the operator $\mc{F}$
computes the right-hand-side of the equation, solves the differential equation,
and finally outputs $v_p(\br t)$. One could employ this formula recursively, i.e.
$v_p^{k+1}=\mc{F}v_p^{k}$. We observed in our example that 
the term $\mb{Q}_{s\alpha}$ becomes noisy even after short times if the 
simulation box is discretized with large spatial steps. This noise 
is received by the partition potential during the propagation, and 
then it is received again by $\mb{Q}_{s,\alpha}$. This feedback process
turns the algorithm unstable. The problem is reminiscent to what occurs in traditional TDDFT when
one wants to find the exchange-correlation potential corresponding to a given 
density, even for only two electrons \cite{LK05}. 
To solve this problem in TDDFT, Ref. \cite{NRL12}
recently suggested an algorithm to control the feedback. They obtained 
encouraging results for a periodic system, but the methodology 
has not been tested for non-periodic systems.

Instead, we found the {\em exact} time-dependent partition potential by using the following optimization procedure: The density
and current density of the total system are found at each time step
using the Crank-Nicolson propagator. (Other propagation methods may
also be used.) A guess is made for the partition potential at the
next unknown time and the fragment wave functions are propagated forward
in time using this guess. (For small time steps the value of the partition
potential at the previous time step works well.) The fragment densities $\left(\left\{ n_{\alpha}\right\} ,\,\left\{ j_{\alpha}\right\} \right)$
are found using these fragment wave functions and added together to form
an approximation to the total densities $\left(n,\, j\right)$.
The errors $n_{\text{err}}=n-n_{\text{exact}}$ and $j_{\text{err}}=j-j_{\text{exact}}$
are computed and the residual $\text{norm}\left(n_{\text{err}}/n_{\text{exact}}\right)+\text{norm}\left(j_{\text{err}}/j_{\text{exact}}\right)$ is
used in the L-BFGS-B optimizer \cite{morales_remark_2011},
with the $L^{2}$ norm. The division by $n_{\text{exact}}$
and $j_{\text{exact}}$ weights the error in the asymptotic
regions to help increase the convergence rate, similar to the weighting used in \cite{peirs_algorithm_2003}.

   
The right panel of Fig. (\ref{f:fig1}) displays the resulting partition potential. The left panel shows 
the total density, along the corresponding fragment densities at the initial time and at 1/4$^{\rm th}$ of a period. 
The importance of memory effects \cite{MB01} is evident from Fig. (\ref{f:fig2}), where the dashed-dotted lines 
labeled ``Adiabatic" show the fragment densities obtained by solving the {\em ground-state} PDFT 
equations for the instantaneous $v(\br t)$ at 1/4$^{\rm th}$ of a period. Clearly, the correct 
partition potential is needed. Only when the electric field strength is reduced by a factor of 10$^3$ (keeping all other parameters fixed) does the adiabatic partition potential produce a molecular density that is visibly indistinguishable from the exact molecular density at time $t\approx T/4$.
Interestingly, the approximation $v_p(\br t)\approx v_p(\br t_0)$ 
(labeled ``Static" in Fig. (\ref{f:fig2})) works qualitatively well for short times, certainly much better than 
the adiabatic approximation. The inset on the right panel of Fig. (\ref{f:fig2}) shows how the static-$v_p$ 
approximation reproduces the correct dipole for short times. 

\begin{figure}[h]
\centering
\begin{tabular*}{1.0\columnwidth}[b]{cc}
\imagetop{\includegraphics[width=0.48\columnwidth]{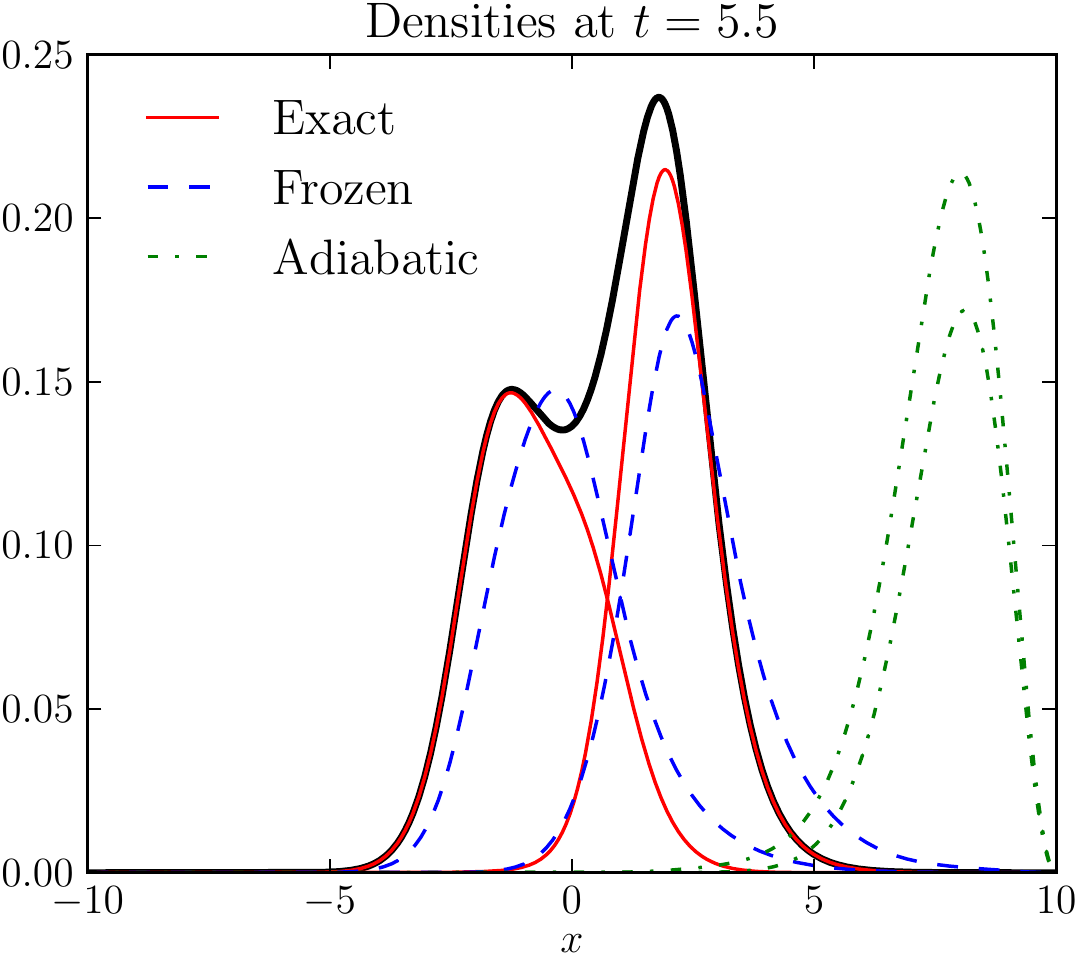}} & 
\imagetop{\includegraphics[width=0.48\columnwidth]{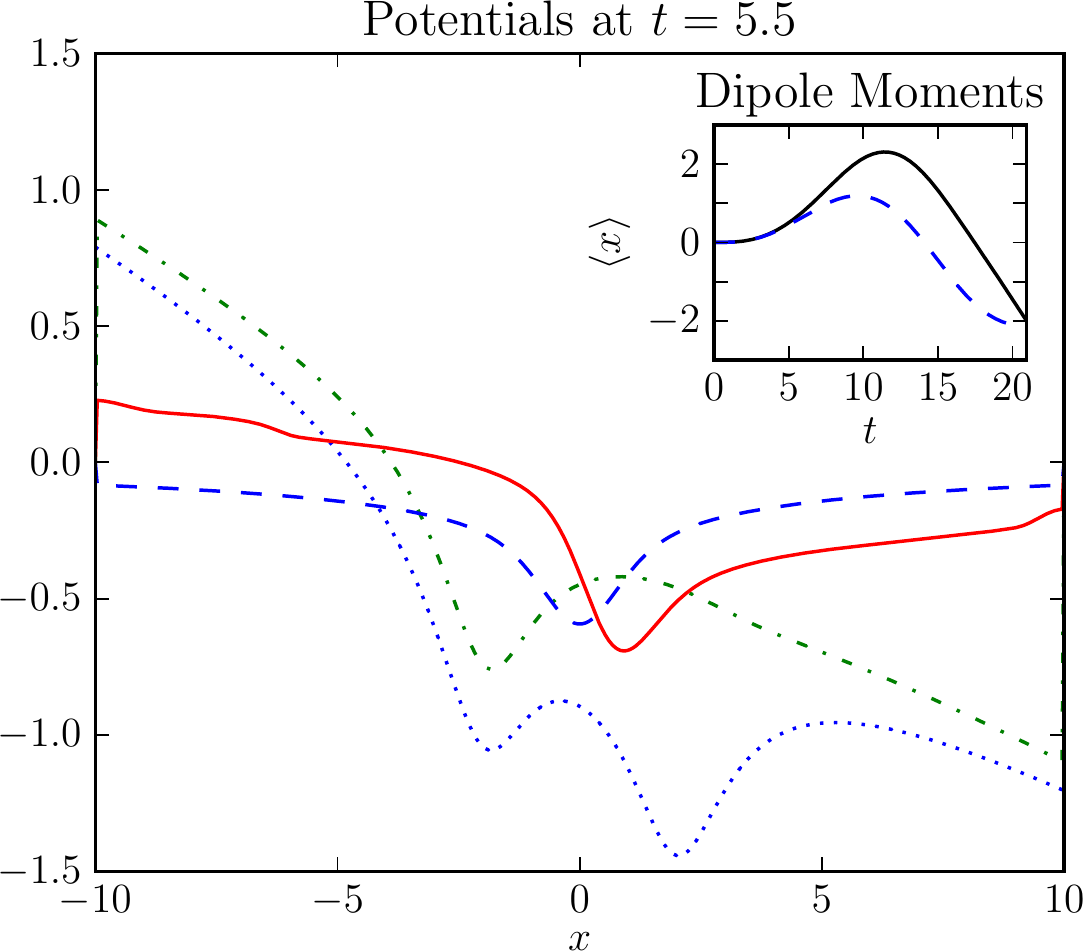}}  
\end{tabular*}
\caption{{\em Left:} Fragment densities at times $t\approx T/4$ along with the
total molecular density (thick solid) calculated: with the {\em exact} $v_p(x t)$ (solid),  static $v_p(x t)=v_p(x t_0)$ (dashed), and adiabatic $v_p(x t)=v_p^{\rm PDFT}[n(x t)]$ (dash-dotted). {\em Right:} Corresponding partition potentials at $t\approx T/4$. Parameters: $l = 1$, $a=1$, $V_{0}=-1$, $E=0.1$, and $\omega=0.3$. The inset compares the dipole moment obtained from the exact (solid) and static approximation to $v_p(x t)$ (dashed). The dotted line indicates the total 
(instantaneous) external potential.}
\label{f:fig2}
\end{figure}



In practice, successful application of our approach to large systems will ultimately rely on the quality of approximations to the time-dependent partition potential. The static approximation might be useful for short times. Furthermore, for problems whose physics is best described by invoking {\em fragments} (such as charge-transfer excitations), we believe that physically-meaningful approximations for $v_p(\br t)$ will be simpler to construct than approximations for the highly non-local exchange-correlation potential and kernel of TDDFT. Work along these lines, as well as on the linear-response formalism, is ongoing.

{\em Acknowledgements:} We acknowledge support from the National Science Foundation CAREER program under grant No.CHE-1149968, and from the Office of Basic Energy Scieneces, U.S. Department of Energy, under grant No. DE-FG02-10ER16196.

\bibliography{refs_ptddft}

\end{document}
%